\documentclass[11pt,a4paper]{article}

\usepackage[numbers]{natbib}

\usepackage{amsmath}
\usepackage{amssymb}
\usepackage{url}
\usepackage{enumerate}
\usepackage{multirow}

\usepackage[ruled,vlined]{algorithm2e}

\usepackage{hyperref}
\usepackage{breakurl}
\usepackage[margin=1.2in]{geometry}

\usepackage{xpatch}

\makeatletter
\xpatchcmd{\algorithmic}
  {\ALG@tlm\z@}{\leftmargin\z@\ALG@tlm\z@}
  {}{}
\makeatother

\usepackage{amsthm}
\newtheorem{example}{Example}
\newtheorem{theorem}{Theorem}

\newtheorem{conjecture}[theorem]{Conjecture}

\title{Calcium: computing in exact real and complex fields}
\date{}
\author{Fredrik Johansson\footnote{Inria Bordeaux and Institut Math. Bordeaux -- \url{fredrik.johansson@gmail.com}}}

\begin{document}

\maketitle

\begin{abstract}
Calcium is a C library for real and complex numbers
in a form suitable for exact algebraic and symbolic computation.
Numbers are represented as elements of
fields $\mathbb{Q}(a_1,\ldots,a_n)$ where
the extensions numbers
$a_k$ may be algebraic or transcendental.
The system combines efficient field operations
with automatic discovery and certification of algebraic relations,
resulting in a practical computational model of $\mathbb{R}$ and $\mathbb{C}$
in which equality is rigorously decidable for a large class of numbers.
\end{abstract}

\section{Introduction}

A field $K$ is said to be \emph{effective}
if its elements can be enumerated
and the operations $\{+, -, \cdot, /, =\}$
are computable.
Examples include the rationals $\mathbb{Q}$, finite fields $\mathbb{F}_q$,
and the algebraic numbers $\overline{\mathbb{Q}}$.

The fields of real and complex numbers~$\mathbb{R}$ and~$\mathbb{C}$
are notably non-effective, even 
when restricted to so-called \emph{computable numbers}
(a real number $x$ is said to be \emph{computable} if
there is a program
which, given $n$, outputs some $x_n \in \mathbb{Q}$ with $|x - x_n| < 2^{-n}$).
The problem is that equality is only semi-decidable: we can in general prove $x \ne y$,
but not $x = y$, as a consequence of the halting theorem.
Nevertheless,
as the example of $\overline{\mathbb{Q}}$ shows,
we can hope for an equality test at least for \emph{some} numbers
within a suitable algebraic framework.

This paper presents Calcium,\footnote{Pronounced ``kalkium'' to distinguish it from
the chemical element. Calcium is free and open source
(LGPL 2.1+) software. The source repository is
\url{https://github.com/fredrik-johansson/calcium} and
the documentation is available at \url{http://fredrikj.net/calcium/}.} a C library
for exact computation in $\mathbb{R}$ and $\mathbb{C}$.
Numbers are represented as elements of fields
$\mathbb{Q}(a_1,\ldots,a_n)$ where the extension numbers $a_k$
are defined symbolically.
The system constructs fields and discovers algebraic relations
automatically,
handling algebraic and transcendental number fields in a unified way.
It is capable of deciding equality
for a wide class of numbers
which includes $\overline{\mathbb{Q}}$ as a subset.
We show a few basic examples, here using a Python wrapper:

\begin{small}
\begin{verbatim}
>>> (pi**2 - 9) / (pi + 3)
0.141593 {a-3 where a = 3.14159 [Pi]}
>>> phi = (sqrt(5)+1)/2; (phi**100 - (1-phi)**100)/sqrt(5)
3.54225e+20 {354224848179261915075}
>>> i**i - exp(pi / (sqrt(-2)**sqrt(2))**sqrt(2))
0
>>> log(sqrt(2)+sqrt(3)) / log(5+2*sqrt(6))
0.500000 {1/2}
>>> erf(4*atan(ca(1)/5) - atan(ca(1)/239)) + erfc(pi/4)
1
>>> -1e-12 < exp(pi*sqrt(163)) - 262537412640768744 < -1e-13
True
\end{verbatim}
\end{small}


In the first example, the field is $\mathbb{Q}(a)$ where $a = \pi$,
and the element is $a-3$.
The numerical approximation ($x \approx 0.141593$) is computed to
desired precision on demand, for example when printing
or evaluating a numerical predicate.
Examples 2--5 were chosen so that the field
of the result simplifies to $\mathbb{Q}$.

Such examples are within the scope of the
expression simplification tools in computer algebra systems like
Mathematica and Maple.\footnote{Open source systems
tend to perform much worse: SymPy, for example,
fails to simplify examples 3--5. Maple curiously also fails to simplify example 3.}
The key difference is that we
work with more structured representations;
we also handle numerical evaluation and predicates rigorously.
Our approach is inspired by
various earlier implementations of $\overline{\mathbb{Q}}$
and by theoretical work on transcendental fields.

This paper is structured as follows. Section~\ref{sect:algebraic}
presents our high-level strategy for exact computation, described at a general level
without reference to low-level implementation details.
Section~\ref{sect:calcium} discusses the architecture of Calcium.
Section~\ref{sect:related} relates our strategy
to earlier work and presents some benchmark results.

\section{Computing in subfields of $\mathbb{C}$}

\label{sect:algebraic}

To simulate $\mathbb{R}$ and $\mathbb{C}$,
we may start with $\mathbb{Q}$ and
lazily extend the field with new numbers $a_k$ as they arise
in computations.
A general way to compute in such extension fields of $\mathbb{Q}$
is in terms of quotient rings
and their fields of fractions (henceforth \emph{formal fields}).

In the following, we assume that $a_1,\ldots,a_n$ is a finite list of complex numbers.
We let $X_1,\ldots,X_n$ denote independent formal variables,
we let $\mu : \mathbb{Q}[X_1,\ldots,X_n] \to \mathbb{C}$ denote
the evaluation homomorphism
induced by the map $X_k \mapsto a_k$,
and we define $$I := \ker \mu = \{ f \in \mathbb{Q}[X_1,\ldots,X_n]: f(a_1,\ldots,a_n) = 0\}$$
as the ideal of all algebraic relations among $a_1,\ldots,a_n$ over $\mathbb{Q}$.

\begin{theorem}
Assume that $I$ is known (in the sense
that an explicit list of generators $I = \langle f_1, \ldots, f_m \rangle$ is known). Then
$$K := \mathbb{Q}(a_1,\ldots,a_n) \;\; \cong \;\; K_{\text{formal}} := \operatorname{Frac}(\mathbb{Q}[X_1,\ldots,X_n] / I)$$
is an effective subfield of $\mathbb{C}$.
\label{thm:effective}
\end{theorem}

\begin{proof}
The isomorphism is obvious.
Decidability of ``='' in the formal
field follows from the fact that we can compute a Gr\"{o}bner basis for~$I$.
Given
formal fractions $\tfrac{p}{q}$ and $\tfrac{r}{s}$ with
$p,q,r,s \in \mathbb{Q}[X_1,\ldots,X_n]$,
we can consequently decide whether $ps \equiv qr \bmod I$.
Indeed, we can also decide whether $q, s \not \equiv 0 \bmod I$ and thereby ensure that
the fractions define numbers in the first place.
\end{proof}

Some easy special cases are worth noting:

\begin{itemize}
\item The trivial field $K = \mathbb{Q}$ (take $n = 0$).
\item Transcendental number fields
$K = \mathbb{Q}(a_1,\ldots,a_n)$
where the numbers $a_1,\ldots,a_n$ are
algebraically independent over $\mathbb{Q}$.
\item Algebraic number fields $K = \mathbb{Q}(a) \; \cong \; \mathbb{Q}[X] / \langle f(X)\rangle$
where $a$ is an algebraic number with minimal polynomial $f$.
\end{itemize}

The general case is a \emph{mixed field} in which the extension
numbers may be algebraic or transcendental and algebraically
dependent or independent in any combination.

\begin{example}
$\mathbb{Q}(\log(i),\pi,i) \cong \operatorname{Frac}(\mathbb{Q}[X_1,X_2,X_3] / I)$
where $I = \langle 2 X_1 - X_2 X_3, \, X_3^2+1 \rangle$.
\end{example}

Theorem~\ref{thm:effective} solves the arithmetic part of computing in
finitely generated
subfields of $\mathbb{C}$, at least up to practical issues such as the complexity
of multivariate polynomial arithmetic and Gr\"{o}bner basis computations.
The crucial assumption made in Theorem~\ref{thm:effective}, however, is that
the ideal $I$ is known.
In general, finding $I$ is an extremely hard problem.
For example, although $\mathbb{Q}(\pi) \cong \mathbb{Q}(X_1)$
and $\mathbb{Q}(e) \cong \mathbb{Q}(X_2)$, it is an open
problem to prove $\mathbb{Q}(\pi,e) \cong \mathbb{Q}(X_1,X_2)$.
There are specific instances where we can prove algebraic
independence (the Hermite-Lindemann-Weierstrass theorem,
Baker's theorem, transcendence of isolated
numbers such as $\Gamma(\tfrac{1}{4})$, results
for $E$-functions~\cite{FR2019}, etc.), but we typically only
have conjectures. Most famous (and implying
$\mathbb{Q}(\pi,e) \cong \mathbb{Q}(X_1,X_2)$ as a special case) is:

\begin{conjecture}[Schanuel's conjecture]
If $z_1,\ldots,z_n$ are linearly independent over $\mathbb{Q}$,
then $\mathbb{Q}(z_1,\ldots,z_n,e^{z_1},\ldots,e^{z_n})$ has
transcendence degree at least $n$ over $\mathbb{Q}$.
\end{conjecture}

Thus, in general, we can only determine $I$ conjecturally.
We address this limitation below in section~\ref{sect:incompleteideal}.

\subsection{Defining extension numbers}

We stress that we cannot simply input $I$ as a way to \emph{define}
the extension numbers $a_1,\ldots,a_n$, since this does not
give enough information about the embedding (that is, $\mu$) in $\mathbb{C}$.
We are not interested in computing in an abstract algebraic
structure but in a concrete model of $\mathbb{R}$
and $\mathbb{C}$ where
we can do (at least) the following:
\begin{itemize}
\item Evaluate the complex conjugation map $z \to \overline{z}$.
\item Evaluate numerical ordering relations ($x < y$, $|x| < |y|$, etc.).
\item Exclude singularities (e.g.\ division by zero) and choose well-defined branches of multivalued functions.
\end{itemize}
We therefore need a symbolic way to define extension numbers
$a_1,\ldots,a_n$ so that they are explicitly (numerically) computable,
and we need to construct $I$ from this symbolic data
rather than vice versa.
The following types of extensions are useful:

\begin{itemize}
\item \emph{Absolute algebraic}: $a$ is a fixed algebraic constant $a \in \overline{\mathbb{Q}}$, for example
$i$, $\sqrt{2}$, $e^{2 \pi i / 3}$, or $[a^5-a-1=0; \, a\approx 1.17]$.
Such a constant can be defined canonically by its
minimal polynomial over $\mathbb{Q}$ together with an isolating
ball for a root.
\item \emph{Relative algebraic}: $a$ is defined by an equation $a^{m/n} = c$ with $c \in \mathbb{C}$, or $P(a) = 0$ with $P \in \mathbb{C}[X]$ together with an isolating complex ball for a root.
\item \emph{Transcendental}: $a$ is a symbolic transcendental (or conjecturally transcendental)
constant ($\pi$, $\gamma$, etc.) or function
($e^z$, $\log(z)$, $z^w$, $\Gamma(z)$, $J_{\nu}(z)$, etc.) evaluated at some point.
\item \emph{Black-box computable}: $a$ is defined by a program for numerical evaluation in ball arithmetic. (We will not be able to prove algebraic relations
except self-relations like $a - a = 0$.)
\end{itemize}

Calcium presently supports extension numbers of the first three types.
We represent algebraic and transcendental extensions
in the usual way as symbolic expressions $f(z_1,\ldots,z_p)$.
The arguments $z_1, \ldots, z_p$ are real or complex numbers
which may belong to different fields, say $z_1 \in \mathbb{Q}(b_1,\ldots,b_r)$,
$z_2 \in \mathbb{Q}(c_1,\ldots,c_s)$, etc.

Each extension number defines a computable number
through recursive numerical evaluation of the symbolic function
and its arguments in arbitrary-precision ball arithmetic.
This is at least true in principle
assuming that we can decide signs at discontinuities.
An important improvement
over many symbolic computation systems
is that we exclude non-numerical extension numbers: for example,
when adding $\log(z)$ as an extension number, we must be
able to prove $z \ne 0$.
We fix principal branches of all multivalued functions.

This is only a starting point: we can imagine other classes of extensions (periods, solutions of
implicit transcendental equations, etc.).
The main point
is not the precise internal classification but the
logical separation between field elements
and extension numbers.

\subsection{Working with an incomplete ideal}

\label{sect:incompleteideal}

As already noted, it is often not feasible
to find the ideal $I$ necessary to define
a formal field isomorphic to $\mathbb{Q}(a_1,\ldots,a_n)$.
Even in cases where all relations in $I$ in principle can be determined,
they may be costly to compute explicitly,
for instance when they involve algebraic
extensions of even moderately high degree.

Fortunately, it is usually sufficient to construct a partial ideal $I_{\text{red}} \subseteq I$.
We call this the \emph{reduction ideal}
since it typically helps keeping expressions partially reduced
(allowing for efficient computations) even if $I_{\text{red}} \ne I$.
The reason why we do not need to ensure $I_{\text{red}} = I$ is that we can use
the evaluation map $\mu$ (implemented in ball arithmetic)
as a witness of nonvanishing for particular field elements.
Algorithm~\ref{alg:decision} provides
a template for evaluating predicates,
given a possibly incomplete reduction ideal $I_{\text{red}}$.

The algorithm uses a \emph{work parameter} $W$. This can
be taken as a numerical precision in bits for step (c), say with
$W_{\text{min}} = 64$ and $W_{\text{max}} = 4096$
and an implied doubling of $W$ on each iteration.
We explain the meaning of ``heuristics with strength $W$'' in step (d) below in section~\ref{sect:ideal}.
If we take $W_{\text{max}} = \infty$ to force a True/False answer,
then termination when $z = 0$ is conditional
on the asymptotic completeness of the methods to find relations in (d).

\begin{algorithm}
\caption{Test if $z = 0$.}

\SetAlgoLined

\textbf{Input:} Extension numbers $a_1,\ldots,a_n$, an element $z \in \mathbb{Q}(a_1,\ldots,a_n)$
represented by a formal fraction $p / q$ with $p, q \in \mathbb{Q}[X_1,\ldots,X_n]$ (such that $\mu(q) \ne 0$),
a reduction ideal $I_{\text{red}} \subseteq I$, and work limits $W_{\text{min}}$, $W_{\text{max}}$.

\textbf{Output:} True (implying $z = 0$), False ($z \ne 0$), or Unknown.

\vspace*{-0.15cm}
\hrulefill
\vspace*{-0.1cm}

\begin{enumerate}

\item For $W = W_{\text{min}}, \ldots, W_{\text{max}}$, do:
\begin{enumerate}
\item If $p \equiv 0 \bmod I_{\text{red}}$, return True.
\item If it can be certified that $I_{\text{red}} = I$, return False.
\item Using ball arithmetic with strength $W$, compute an enclosure $E$ with $\mu(p) \in E$. If $0 \not\in E$, return False.
\item Using heuristics with strength $W$, attempt to find and prove a new set of relations $J$ with $J \subseteq I$, and set $I_{\text{red}} \gets I_{\text{red}} \cup J$. (See Algorithm~\ref{alg:ideal}.)
\end{enumerate}
\item Return Unknown.
\end{enumerate}

\label{alg:decision}
\end{algorithm}

Step (b) is applicable, for example, in simple algebraic or transcendental number
fields such as $\mathbb{Q}(\sqrt{2})$ and $\mathbb{Q}(\pi)$.

In step (d), we may try to find a general relation
for $a_1,\ldots,a_n$,
or we may attempt to prove $\mu(p) = 0$ directly
and take $J = \langle p \rangle$. The latter is sometimes easier.
For example, if $a_1,\ldots,a_n$ are algebraic extension numbers,
it is often cheaper to compute the minimal polynomial
specifically for $z$ than to compute all of $I$.

It is an implementation detail whether we cache the updated
ideal $I_{\text{red}}$ for future use
in the same field after exiting Algorithm~\ref{alg:decision}.

\subsection{Constructing the ideal}

\label{sect:ideal}

We will now describe a practical strategy to construct a reduction ideal
$I_{\text{red}}$ for a given field $K = \mathbb{Q}(a_1,\ldots,a_n)$.

Which relations are interesting to include, and when?
The minimalist solution is that we set $I_{\text{red}} = \{\}$
when we construct $K$, and only populate $I_{\text{red}}$
lazily in Algorithm~\ref{alg:decision}.
The maximalist solution is to ensure
$I_{\text{red}} = I$ up front.
There is a tradeoff: on one hand, we want to capture as much of 
the true ideal $I$ as possible so that testing equality is trivial
and so that there is minimal expression swell in computations.
On the other hand, we do not want to waste time
finding potentially useless relations and computing Gr\"{o}bner bases every time we construct
a field.
As in Algorithm~\ref{alg:decision}, it is useful to
make the effort dependent on a work parameter $W$
controlling numerical precision, choice of heuristics,
and so forth.

Algorithm~\ref{alg:ideal} implements a smorgasbord of methods for
finding relations, most of which involve
searching for (linear) integer relations.
We recall that an integer relation between
complex numbers $a_1,\ldots,a_n$ is a tuple $(m_1, \ldots, m_n)$ with some $m_k \ne 0$
such that
$$m_1 a_1 + \ldots m_n a_n = 0, \quad m_i \in \mathbb{Z}.$$

The LLL algorithm can be used to compute a basis
matrix for all integer relations among
a finite list of numbers; see for example Algorithm~7.13 in \cite{Kau2005}.
More precisely, LLL finds a basis of candidate
relations which may or may not be correct.
We are guaranteed to find all integer relations
as $W \to \infty$ where $W$ is the numerical precision,
but we have to use exact computations to certify or reject the relations
obtained at a fixed finite $W$. Since the certifications can
be expensive, it is useful to make them dependent on $W$ (for
example, limiting bit sizes of field elements in recursive computations).

\begin{algorithm}
\caption{Construct ideal of algebraic relations.}

\SetAlgoLined

\textbf{Input:} Extension numbers $a_1,\ldots,a_n$, a work parameter $W$.

\textbf{Output:} A reduction ideal $I_{\text{red}} \subseteq I$ for $\mathbb{Q}(a_1,\ldots,a_n)$.

\vspace*{-0.15cm}
\hrulefill
\vspace*{-0.05cm}

Initialize $I_{\text{red}} \gets \{ \}$.
Depending on $W$, run a subset of A-F:

\begin{enumerate}[{A}]
\item \textbf{Direct algebraic relations}. For absolute or relative algebraic
extensions $a_k$, add the defining relations to $I_{\text{red}}$.

\item \textbf{Vieta's formulas}. For algebraic extensions $a_k$ that are conjugate roots
of the same polynomial, add the interrelations defined by Vieta's formulas to $I_{\text{red}}$.

\item \textbf{Log-linear relations}. Let $L$ denote the set of extension numbers of the form $a_k = \log(z_k)$,
along with $\pi i$ if available.
Use LLL with precision $W$ to search for relations $\sum_j m_j \log(z_j) = 0$ or $m_0 (2 \pi i) + \sum_j m_j \log(z_j) = 0$.
Attempt to certify each candidate relation:
\begin{itemize}
\item Compute an enclosure of $\tfrac{1}{2 \pi i} \sum_j m_j \log(z_j)$ and verify that it contains a unique integer.
\item Attempt to prove $\prod_j {z_j}^{m_j} = 1$ using Algorithm~\ref{alg:decision} (using exact recursive computations in the fields of the arguments $z_j$).
\item If both certification steps succeed, update the ideal with $I_{\text{red}} \gets I_{\text{red}} \cup \langle m_0 (2 \pi i) + m_1 a_1 + \ldots + m_n a_n \rangle$.
\end{itemize}
\item \textbf{Exp-multiplicative relations}.
Let $E$ denote the set of extension numbers of the form $a_k = \sqrt{z_k}$, $a_k = {z_k}^{m/n}$, $a_k = {z_k}^{w_k}$, $a_k = e^{z_k}$ or $a_k \in \overline{\mathbb{Q}}$.
Search for potential multiplicative relations $\prod_j {a_j}^{m_j} = 1$ using LLL applied to $\log(E)$ and certify the candidate relations through exact recursive computations similarly to the log-linear case.

\item \textbf{Special functions}. Update $I_{\text{red}}$ with relations
resulting
from functional equations and connection formulas such as $\Gamma(z+1) = z \Gamma(z)$
or $\operatorname{erf}(z) = -\operatorname{erf}(-z) = -i \operatorname{erfi}(i z)$.
Candidate relations can be found by numerical comparison of function arguments
and certified through exact recursive computations.

\item \textbf{Algebraic interrelations}. Use resultants or LLL (followed by certification
using resultants) to search for linear (or bilinear, etc.) relations
among algebraic extensions.

\end{enumerate}

\label{alg:ideal}
\end{algorithm}

Algorithm~\ref{alg:ideal} will only find relations
that are expressible in terms of the given
$a_1,\ldots,a_n$.
For example, to add the relation
for a square root extension $a_k = \sqrt{z}$ in step A,
we need to be able to express $z$ in terms
of $K' = \mathbb{Q}(a_1,\ldots,a_{k_1},a_{k+1},a_n)$
as a formal fraction $f/g$ with $\mu(f / g) = z$.
The relation is then $\langle g^2 X_k^2 - f^2 \rangle$.
If $z$ cannot be expressed in $K'$, then $a_k$
behaves like a transcendental number within the present field.
However, it is usually desirable to make $z$ part of the field
so that ideal reduction automatically produces $(\sqrt{z})^2 \rightarrow z$.
One possibility is that we always adjoin $z$ (or $b_1,\ldots,b_m$
such that $z \in \mathbb{Q}(b_1,\ldots,b_m)$) to the field
where we create $\sqrt{z}$.
An alternative is to modify Algorithm~\ref{alg:ideal} so that it
can append new extension numbers to the existing field
whenever it may help simplifications.

We will not attempt to prove the completeness of
Algorithm~\ref{alg:ideal} for any particular sets of numbers here
(see section~\ref{sect:related} for a few remarks).
We are constrained by the requirement
that potential relations have to be certifiable:
it makes no sense to look for a hypothetical relation that we will
not be able to prove (say, $m \pi + n e = 0$ with $m, n \in \mathbb{Z}$).\footnote{We can imagine an optional ``nonrigorous mode''
similar to the algorithm in \cite{BBK2014} which looks for
numerical integer relations and uses them to simplify
symbolic expressions without guaranteeing correctness.}

\subsection{Choosing extension numbers}

We have so far assumed that the extension numbers $a_1,\ldots,a_n$
are given. We usually have a great deal of freedom to
choose the form of extension numbers to represent a given field $K$. The following are some possible
transformations that either generate a new representation of $K$ itself,
generate a larger field $K' \supseteq K$, or generate a subfield or overlapping field:

\begin{itemize}
\item \emph{Normalization}: replacing an extension number by a simpler (by some measure) generator of the same field.  Example: $\mathbb{Q}(-\tfrac{5}{3}\sqrt{8}) \to \mathbb{Q}(\sqrt{2})$, $\mathbb{Q}(e^{-\pi}) \to \mathbb{Q}(e^{\pi})$.
\item \emph{Pruning}: removing redundancy. Ex.: $\mathbb{Q}(-\sqrt{2},\sqrt{2}) \to \mathbb{Q}(\sqrt{2})$.
\item \emph{Unification}: replacing extensions by a common generator. Ex.: $\mathbb{Q}(\sqrt{2},\sqrt{3}) \to \mathbb{Q}(\sqrt{2}+\sqrt{3})$, $\mathbb{Q}(\pi^{1/2}, \pi^{1/3}) \to \mathbb{Q}(\pi^{1/6})$.
\item \emph{Specialization}: simplifying special cases. Ex.: $\mathbb{Q}(e^0) \to \mathbb{Q}$, $\mathbb{Q}(e^{\log(z)}) \to \mathbb{Q}(z)$ and $\mathbb{Q}(\log(e^z)) \to \mathbb{Q}(z, \pi, i)$.
\item \emph{Atomization}: rewriting an extension in terms of more ``atomic'' parts. Ex.: $\mathbb{Q}(\sqrt{2}+\sqrt{3}) \to \mathbb{Q}(\sqrt{2},\sqrt{3})$,
$\mathbb{Q}(e^{x+y}) \to \mathbb{Q}(e^x, e^y)$, and $\mathbb{Q}(\log(xy)) \to \mathbb{Q}(\log(x), \log(y), \pi, i)$.
\item \emph{Function replacement}: rewriting a function in terms of a different function or combination of functions. Ex.: $\mathbb{Q}(\sin(x)) \to \mathbb{Q}(e^{ix}, i)$, $\mathbb{Q}(e^{x+yi}) \to \mathbb{Q}(e^x, \cos(y), \sin(y), i)$.
\end{itemize}

The problem of choosing appropriate extension numbers arises in various situations:

\begin{itemize}
\item Evaluating functions and solving equations: for example, given $z$, construct a field to represent $\sqrt{z}$ or $e^z$.
\item Merging fields, especially for arithmetic: given $z_1 \in K_1 = \mathbb{Q}(a_1,\ldots,a_n)$ and $z_2 \in K_2 = \mathbb{Q}(b_1,\ldots,b_m)$, compute a field $K_3$ containing $z_3 = z_1 \circ z_2$ where $\circ$ is an arithmetic operation.
\item Simplifying a single element (or finite list of elements): given $z \in K$, construct $K' \subseteq K$ with $z \in K'$ that is better suited for deciding a predicate, user output, numerical evaluation, etc.
\end{itemize}

We can attempt to set reasonable defaults, but a useful system
should probably allow the user to make intelligent choices.
It is very difficult to define meaningful canonical forms for general
symbolic expressions, and the optimal
form often depends on the application~\cite{Mos1971,Car2004}.
A classical problem is whether it makes sense to expand $(\pi+1)^{1000}$ in $\mathbb{Q}(\pi)$
or whether the result should be represented in $\mathbb{Q}((\pi+1)^{1000})$
(in Calcium, this is configurable).
Although atomization intuitively simplifies extensions,
having more variables
can slow down the task of constructing
the ideal and performing operations in the formal field,
and in any case the choice of ``atoms'' is often somewhat arbitrary.

We follow a conservative approach in Calcium so far: merging fields simply
takes the union of the generators,
evaluating functions or creating algebraic numbers
only normalizes or specializes in trivial cases,
and automatic pruning is mainly done to demote rational numbers
to $\mathbb{Q}$. In the future, we intend to implement different
behaviors and make them configurable, allowing the user to choose different
``flavors'' of arithmetic (for example, always unifying algebraic numbers to
a single extension, always separating complex numbers into
real and imaginary parts, etc.).

Algorithm~\ref{alg:ideal} 
is notably missing heuristics for
trigonometric functions
and complex parts (real part, imaginary part, sign,
absolute value).
We can write trigonometric functions and their inverses
in terms of complex exponentials and logarithms,
and complex parts in terms of algebraic operations
and recursive complex conjugation or separation of real
and imaginary parts,
but this is not always appropriate, particularly when we end up using
complex extensions to describe a real field.
We leave this problem for future work.

\subsection{On formal field arithmetic}

We conclude this section with
some practical comments
about implementing formal fields
$\operatorname{Frac}(\mathbb{Q}[X_1,\ldots,X_n] / I)$.

\subsubsection{Normal forms of fractions}

When computing in formal fraction fields,
we face a difficulty
which does not arise when merely considering
quotient rings
$\mathbb{Q}[X_1,\ldots,X_n] / I$:
a formal fraction $p / q$ need not be
in a joint canonical form even if $p$ and $q$ are in canonical
form with respect to $I$.\footnote{For the present discussion, it does
not matter whether we have $I_{\text{red}} = I$.}
A simple example is that $1/\sqrt{2} = \sqrt{2}/2$.
This is harmless for deciding equality since
reduction by $I$ will give a zero numerator of
$p / q - r / s = (ps - rq)/(pq)$ for equivalent fractions $p/q$ and $r/s$.
However, $p$ and $q$ can have nontrivial
common content in $\mathbb{Q}[X_1,\ldots,X_n] / I$
even if they are coprime in $\mathbb{Q}[X_1,\ldots,X_n]$,
and failing to remove such content can result in expression swell.
This problem manifests itself, for example, in Gaussian elimination.

In special cases, it is possible to find
content by computing polynomial GCDs over an algebraic number field
instead of over $\mathbb{Q}$.
Monagan and Pearce \cite{Mon2006} provide
an algorithm that solves the general problem of simplifying
fractions modulo an arbitrary (prime) ideal.
Their algorithm uses Gr\"{o}bner bases over modules.
We have not yet implemented this method in Calcium,
and only remove content in $\mathbb{Q}[X_1,\ldots,X_n]$
from formal fractions
(except in the special case of simple algebraic number fields,
where we compute a canonical form by rationalizing the denominator).

\subsubsection{Orderings}

The choice of \emph{monomial order} (\emph{lex}, \emph{deglex}, \emph{degrevlex}, etc.),
for multivariate polynomials in formal fields
can have a significant impact on efficiency and simplification power.
Closely related is the \emph{extension number order}: we typically
want to sort the extension numbers
in order of decreasing
complexity $a_1 \succ a_2 \succ \ldots \succ a_n$
for lexicographic elimination.
The notion of complexity is somewhat
arbitrary, but typically for any symbolic function $f$ and
any $z$, we want $f(z) \succ z$. For a discussion
of the problem of ordering symbolic expressions, see \cite{Mos1971,Car2004}.

Overall, \emph{lex} monomial ordering often seems to perform best
due to its tendency to completely eliminate extension numbers of higher complexity,
and it is used by default in Calcium,
although degree orders sometimes lead to cheaper
Gr\"{o}bner basis computations and overall simpler polynomials.
Calcium currently uses a hardcoded comparison function for
extension numbers, but we intend
to make it configurable or context-dependent.
A more sophisticated system might use heuristics
to choose an appropriate extension number order and monomial order (including
weighted and block orders) for each extension field.

\section{Architecture of Calcium}

\label{sect:calcium}

In this section, we describe the design of Calcium
as a library and discuss certain low-level implementation
aspects.

We chose to implement Calcium as a C library
to minimize dependencies.
Calcium includes a simple, unoptimized Python wrapper (using \texttt{ctypes})
intended for easy testing.

Calcium depends on
Arb~\cite{Joh2017} for arbitrary-precision ball arithmetic,
Antic~\cite{Har2015} for arithmetic in algebraic number fields,
and Flint~\cite{Har2010} for rational numbers,
multivariate polynomials and other functionality such as factoring and LLL.
A central idea behind Calcium is to leverage
these libraries for fast in-field arithmetic combined with rigorous evaluation of numerical predicates.
At present, we use a naive implementation of Buchberger's algorithm for
Gr\"{o}bner basis computation, which can be a severe bottleneck.

\subsection{Numbers and context objects}

The main types in Calcium are context objects (\texttt{ca\_ctx\_t})
and numbers (\texttt{ca\_t}). The context object
is the parent object for a 
``Calcium field'', representing a lazily expanding subset of $\mathbb{C}$.
It serves two purposes: it holds a cache of extension numbers
and fields, and it specifies work limits and other settings.
Examples of configurable parameters in the context object include:
the maximum precision for numerical evaluation,
precision for LLL,
the degree of algebraic number fields,
use of Gr\"{o}bner bases, use of Vieta's formulas,
the maximum $N$ for in-field expansion of $(x+y)^N$.
The user may create different contexts configured for different purposes.

The main Calcium number type, \texttt{ca\_t}, holds a pointer
to a field~$K$ and an element of $K$.
As in GMP, Flint and Arb, \texttt{ca\_t} variables have mutable semantics
allowing efficient in-place operations.
Internally, \texttt{ca\_t} uses one of three possible storage types
for field elements:
\begin{itemize}
\item A Flint \texttt{fmpq\_t} if $K = \mathbb{Q}$.
\item An Antic \texttt{nf\_elem\_t} if $K = \mathbb{Q}(a), a \in \overline{\mathbb{Q}}$ is a simple algebraic number field. There are two storage sub-types: Antic uses a specialized inline representation for quadratic fields.
\item A rational function \texttt{fmpz\_mpoly\_q\_t} (implemented as a pair of Flint multivariate polynomials \texttt{fmpz\_mpoly\_t}) if $K$ is a generic (multivariate or non-algebraic) field. Arithmetic in this representation relies on the Flint functions for multivariate arithmetic, GCD, and ideal reduction. Some functions also use Flint's multivariate polynomial factorization.
\end{itemize}

Caching field data in a context object rather than storing the
complete description of a field in each \texttt{ca\_t} variable
is essential for performance: creating new fields can be
expensive; repeated operations and creation of elements within a field
should be cheap.

Calcium is threadsafe as long as two threads never
access the same context object simultaneously. The user can most easily
ensure this by creating separate context objects for each thread.
For fine-grained parallelism, it is most convenient to
convert elements to simpler types such as polynomials.
Some of the underlying polynomial and matrix operations
are parallelized in Flint.

\subsection{Fields and extension numbers}

Separate types are used internally
in the recursive construction of fields.
A \texttt{ca\_ext\_t} object
defines an extension number.
This can be an algebraic number (see below)
or a symbolic constant or function of the form $f(x_1,\ldots,x_n)$
where $x_k$ are \texttt{ca\_t} arguments and $f$ is a builtin
symbol (\texttt{Pi}, \texttt{Exp}, etc.).
A \texttt{ca\_field\_t} object represents a field
$\mathbb{Q}(a_1,\ldots,a_n)$
as an array of pointers to the \texttt{ca\_ext\_t} objects
$a_1,\ldots,a_n$.
Field objects also store computational
data such as the reduction ideal.
Unlike field elements, fields and extension numbers are in principle immutable,
but cached data may be mutated internally: for example, extension numbers
cache Arb enclosures
and update this data when the internal working precision is increased.

The \texttt{ca\_ctx\_t} context object stores
\texttt{ca\_ext\_t} and \texttt{ca\_field\_t} objects
without duplication
in hash tables for fast lookup.
Presently, the context object holds on to all data
until it is destroyed by the user. For applications where memory usage could become
an issue, an improvement would be to add
automatic garbage collection.

\subsection{Canonical algebraic numbers}

Calcium contains
a type \texttt{qqbar\_t} which represents an algebraic number by its 
minimal polynomial over~$\mathbb{Q}$ together with an
isolating complex interval for a root.
Elements of $\overline{\mathbb{Q}}$ are thus represented canonically,
whereas a \texttt{ca\_t} allows many different representations.

The \texttt{qqbar\_t} type is used internally to represent
absolute algebraic extension numbers
and as a fallback
to simplify or test equality of algebraic numbers when Algorithm~\ref{alg:ideal}
fails to find a sufficient reduction ideal.
We thus have a complete test for equality in $\overline{\mathbb{Q}}$.

An arithmetic operation in the \texttt{qqbar\_t} representation involves three steps:
resultant computation (using the BFSS algorithm~\cite{Bos2006}), factoring in $\mathbb{Z}[x]$ (using the van Hoeij algorithm in Flint),
and maintenance of the root enclosure (using interval Newton iteration and other methods based on Arb).
Factoring
nearly always dominates, and this is usually much more expensive
than a \texttt{ca\_t} operation in a fixed number field.
Nevertheless, \texttt{qqbar\_t} performs better than \texttt{ca\_t}
in some situations
and does not require a context object, making it a useful implementation of $\overline{\mathbb{Q}}$
in its own right.

\subsection{Polynomials and matrices}

Calcium provides types \texttt{ca\_poly\_t} and \texttt{ca\_mat\_t}
for representing dense univariate polynomials and matrices over $\mathbb{R}$
or $\mathbb{C}$.
They support arithmetic, predicates,
polynomial GCD and squarefree factorization (using the Euclidean algorithm),
matrix LU factorization, rank and inverse (using ordinary and fraction-free Gaussian elimination),
determinant and characteristic polynomial,
and computing roots or eigenvalues with
multiplicities. Most algorithms are basic,
and optimization could be an interesting future project.

\subsection{Predicates and special values}

There are two kinds of predicate functions: structural
and mathematical. The structural version
of the predicate $x = y$ for \texttt{ca\_t} variables
asks whether $x$ and $y$ contain identically
represented elements
of the same field. This is cheap to check and
gives True/False.
The mathematical predicate asks whether $x$ and $y$ represent
the same complex number.
This is potentially expensive, if not undecidable, and gives True/False/Unknown.
We anticipate that applications using Calcium from a high-level
language will prefer to return True/False and throw an exception for Unknown.
The included Python wrapper does precisely this.
We illustrate matrix inversion:

\begin{small}
\begin{verbatim}
>>> ca_mat([[1, pi], [0, 1/pi]]).inv()          # nonsingular
[1, -9.86960 {-a^2 where a = 3.14159 [Pi]}]
[0,     3.14159 {a where a = 3.14159 [Pi]}]
>>> ca_mat([[pi, pi**2], [pi**3, pi**4]]).inv()    # singular
  ...
ZeroDivisionError: singular matrix
>>> ca_mat([[1, 0], [0, 1-exp(exp(-10000))]]).inv() # unknown
  ...
ValueError: failed to prove matrix singular or nonsingular
>>> ca_mat([[pi, pi**2], [pi**3, pi**4]]).det()
0
>>> ca_mat([[1, 0], [0, 1-exp(exp(-10000))]]).det()
0e-4342 {-a+1 where a = 1.00000 [Exp(1.13548e-4343 {b})],
                    b = 1.13548e-4343 [Exp(-10000)]}
>>> _ == 0
  ...
ValueError: unable to decide predicate: equal
\end{verbatim}
\end{small}

(The third matrix can be inverted by raising the precision limit.)

Like IEEE 754 floating-point arithmetic, 
\texttt{ca\_t} also supports \emph{nonstop computing}
and allows representing non-finite limiting values.
To this end, the \texttt{ca\_t} type actually represents a set $\mathbb{C}^{**} \supset \mathbb{C}$
comprising numbers as well as
various special values:
\emph{unsigned infinity} ($\tilde \infty$),
\emph{signed infinities} ($c \cdot \infty$),
\emph{undefined} ($\mathfrak{u}$)
and \emph{unknown} ($?$).
Formally, $\mathbb{C}^{*} = \mathbb{C} \cup \{\tilde \infty \}
\cup \{c \cdot \infty : |c| = 1\} \cup \{ \mathfrak{u} \}$ and
$\mathbb{C}^{**} = \mathbb{C}^{*} \cup \{ ? \}$.
The sets $\mathbb{C}^{*}$ and $\mathbb{C}^{**}$ are
easily implemented on top of $\mathbb{C}$:
the \texttt{ca\_t} type encodes
special values by two bits in the field pointer.
Unlike IEEE 754,
we disambiguate two NaN types
(mathematical and computational indeterminacy $\mathfrak{u}$ and ?),
we do not distinguish between $-0$ and $+0$,
and complex infinities are represented in polar rather than rectangular form
(we take $\infty + 2i = \infty$).

$\mathbb{C}^{*}$ is a \emph{singularity closure} of $\mathbb{C}$ in which
we can extend partial functions $f : \mathbb{C} \setminus \operatorname{Sing}(f) \rightarrow \mathbb{C}$ to
total functions $f : \mathbb{C}^{*} \rightarrow \mathbb{C}^{*}$.
For example: $1 / 0 = \tilde \infty$, $\log(0) = -\infty$, and $0 / 0 = \infty - \infty = \mathfrak{u}$.
The definitions are simply a matter of convenience
(this particular choice of singularity closure is largely copied from Mathematica).

$\mathbb{C}^{**}$ is a \emph{meta-extension} of $\mathbb{C}^{*}$
in which algorithms can be guaranteed to terminate.
The meta-value (?) represents an undetermined element of $\mathbb{C}^{*}$.
For example, $1 / x$ evaluates to ? if Calcium cannot decide whether $x = 0$
since the value could be either a number or $\tilde \infty$.
Logical predicates on $\mathbb{C}$ extend to logical predicates on $\mathbb{C}^{*}$
($\infty = \infty$, $\mathfrak{u} = \mathfrak{u}$ and $\mathfrak{u} \ne 3$ are all True)
while predicates on $\mathbb{C}^{**}$ are tripled-valued
($? = 3$, $? = \infty$, $? = \mathfrak{u}$ and $? =\;?$ are all Unknown).

We stress the distinction between numbers and special values:
a \texttt{ca\_t} is explicitly
a number, explicitly a singularity (infinity or undefined),
or explicitly unknown.
It is thus easy to restrict usage strictly to $\mathbb{C}$,
in contrast to many symbolic
computation systems where expressions that represent numbers
are syntactically indistinguishable from expressions that are
singular or undefined.

\section{Related work and benchmarks}

\label{sect:related}

The strategy we have discussed
is essentially an attempt to unify several existing
paradigms for exact computation: \emph{effective real numbers},
\emph{symbolic expressions}, \emph{(embedded) number fields},
and \emph{(embedded) quotient rings}. The novelty is not the combination
of functionality (any general-purpose
computer algebra system supports the requisite operations), but the implementation form and interface.

\subsection{Effective numbers}

There are many implementations of ``effective'' or ``computable'' numbers
which construct a symbolic representation to permit
lazy numerical evaluation to arbitrary precision, \emph{without} the
ability to decide equality~\cite{vdH2006,vdH2006b,Mul2001,Yu2010}.
Our representation is more powerful, but likely inferior
if the only goal is numerical evaluation:
symbolic computations are often slower,
and the rewritten expressions can have worse numerical properties
(they will sometimes be better). We rather view Calcium as a second
option to try if a direct numerical evaluation fails
because it stumbles on an exact comparison.

\subsection{Symbolic and algebraic systems}

Most computer algebra systems arguably belong to one of
two paradigms.
\emph{Algebraic} systems
(Singular~\cite{DGPS2019}, Magma~\cite{BCP1997}, Pari/GP~\cite{Par2019}, SageMath~\cite{Sag2020}, Nemo/Hecke~\cite{Fie2017}, etc.)
are designed for computation
in definite algebraic structures, favoring strong data invariants.
\emph{Symbolic} ones (Mathematica, Maple, Maxima~\cite{Max2020}, SymPy~\cite{Meu2017}, etc.),
are designed around more heuristic manipulation of free-form symbolic expressions.
Roughly speaking, algebraic systems tend to prefer $(x+1)^{100}$ in expanded (normal) form
and view it as an element of a particular ring such as $\mathbb{Q}[x]$,
while symbolic systems tend to leave it unexpanded
and unassociated with a formal algebraic structure.
Calcium is an attempt to
provide a more algebraic package
for functionality for real and complex numbers
previously only found in symbolic systems.
The algebraic approach has benefits for performance and correctness,
although we lose some flexibility:
we notably give up most superficial manipulation of rational expressions
(\emph{expand}, \emph{combine}, \emph{apart}, \emph{factor}, etc.), for better or worse.

Calcium is not a general-purpose expression simplifier
like the \texttt{simplify} or \texttt{FullSimplify} routines in systems
like Maple and Mathematica,
which combine many heuristics.
A roundtrip \emph{expr} $\rightarrow$ Calcium $\rightarrow$ \emph{expr}
can be a useful part in the toolbox of such a simplifier,
but will often have to be applied selectively.
In fact, Calcium grew out of code to manipulate and test
symbolic expressions in FunGrim~\cite{Johansson2020},
with the view of having a middle layer between symbolic expressions
and polynomial and ball arithmetic.

\subsection{Algebraic numbers}

Computing in $\overline{\mathbb{Q}}$ is a well-studied problem
which admits multiple approaches~\cite{Coh1996,Yu2010}.
A generally useful principle is to rely on arithmetic
in fixed number fields for efficiency.
Calcium is partly inspired by
Sage's \texttt{QQbar}, which uses a hybrid representation:
an algebraic number exists either as an
element of a number field $\mathbb{Q}(a)$ or
an unevaluated symbolic expression.
Numerical values are tracked rigorously using interval arithmetic.
A comparison that cannot be resolved numerically
forces a simplification to an absolute field.

Sage's approach has two problems: unevaluated symbolic
expressions fail to capture arithmetic simplifications,
and combining extensions to a single
absolute number field can be costly.
Calcium's multivariate representation often avoids
costly simplifications.
A simple test case is to compute $x = \sqrt{2} + \sqrt{3} + \ldots + \sqrt{p_n}$
and check $x-(x-1)-1 = 0$.
For $n = 7$ this takes 2200~s in Sage and 0.003~s in Calcium, including the
time to set up and clear a context object (Calcium takes 0.00007~s with the fields already cached).\footnote{For a more complex test problem that Calcium handles easily
where Sage struggles unless carefully guided, see: \url{https://ask.sagemath.org/question/52653}}

Calcium is also inspired by the \emph{algebraically closed field}
in Magma,
which uses multivariate quotient rings
$\mathbb{Q}[X_1,\ldots,X_n] / I$
that grow automatically~\cite{Ste2002,Ste2010}.
Magma's \texttt{ACF} is not actually a satisfactory implementation of $\overline{\mathbb{Q}}$
for our purposes
because it does not define the embedding of polynomial roots into $\mathbb{C}$:
choices that depend on permuting roots are made arbitrarily by the system
and cannot be predicted by the user.
On the other hand, Magma uses more sophisticated methods
in its ideal construction than we currently do for algebraic extensions; an interesting
future project would be to integrate some of its techniques with our system.

\subsection{Elementary numbers}

The next interesting structure after $\overline{\mathbb{Q}}$
is the field of \emph{elementary numbers}.
To be precise, there are two common definitions of such a field:
the \emph{exp-log field} $\mathbb{E}$ is the closure of $\mathbb{Q}$
with respect to exponential and logarithmic extensions $e^z$ and $\log(z)$ ($z \ne 0$),
while the \emph{Liouvillian field} $\mathbb{L}$ is
the closure of $\mathbb{Q}$ with respect to algebraic,
exponential and logarithmic extensions.\footnote{Clearly $\overline{\mathbb{Q}} \ne \mathbb{E}$, $\mathbb{E} \subseteq \mathbb{L}$ and $\overline{\mathbb{Q}} \subseteq \mathbb{L}$,
but it is unknown if $\mathbb{E} = \mathbb{L}$ and if $\overline{\mathbb{Q}} \subset \mathbb{E}$.~\cite{Cho1999}}

Richardson~\cite{RF1994, Ric1992, Ric1995, Ric1997, Ric2007, Ric2009} has
constructed a decision procedure for equality in $\mathbb{E}$ and $\mathbb{L}$
using computations in towers of extensions over $\mathbb{Q}$,
which always succeeds
if Schanuel's conjecture is true
and will loop forever when given a counterexample.
The surface part of such a decision procedure is essentially Algorithm~\ref{alg:decision},
in which we iteratively attempt to either prove inequality
or find an algebraic relation that implies equality.
Assuming Schanuel's conjecture, it can be shown that any relation between
elementary numbers must result from a combination of
log-linear relations, exp-multiplicative relations,
and relations resulting from the identical
vanishing of algebraic functions (for example, $\sqrt{(\log(2))^2} - \log(2) = 0$,
due to the identical vanishing of $\sqrt{x^2}-x$ on the local branch).
Algorithm~\ref{alg:ideal} is inspired by Richardson's algorithm,
but incomplete: it will find logarithmic and exponential
relations, but only
if the extension tower is flattened
(in other words, we must avoid
extensions such as $e^{\log(z)}$ or $\sqrt{z^2}$),
and it does not handle all algebraic functions.

Much like the Risch algorithm,
Richardson's algorithm has apparently never been implemented fully.
We presume that Mathematica and Maple use similar
heuristics to ours, but the details
are not documented~\cite{Car2020}, and we do not know to what
extent True/False answers are backed up by a rigorous certification
in those system.

A practical difficulty when comparing numbers involving
elementary functions
is that extremely high precision may be needed
to distinguish nested exponentials numerically (as an
example, consider $\exp(-e^{-e^{N}}) \ne 1$).
This problem can be overcome using asymptotic expansions~\cite{vdH1995}.
We have not yet investigated such methods.
The recent work \cite{Boe2020} uses irrationality criteria
to prove some inequalities,
but this is only applicable in very restricted cases.

\subsection{Miscellaneous examples}

This section is short due to page limits.
For code and additional benchmarks, we refer to
example programs included with Calcium.\footnote{See \url{http://fredrikj.net/calcium/examples.html}
and \url{http://fredrikj.net/blog/2020/09/benchmarking-exact-dft-computation/} for the complete data for the DFT benchmark.}

\subsubsection{Exact DFT}

As a test of basic arithmetic and simplification,
we check $\textbf{x} - \operatorname{DFT}^{-1}(\operatorname{DFT}(\textbf{x})) = \textbf{0}$
where $\textbf{x} = (x_n)_{n=0}^{N-1}$ and $\operatorname{DFT}(\textbf{x}) = \sum_{k=0}^{N-1} \omega^{-kn} x_k$
with $\omega = e^{2 \pi i / N}$. For this benchmark, we evaluate
the DFT by naive $O(N^2)$ summation (no FFT).
We test six input sequences exhibiting both algebraic and transcendental numbers.

We compare with Sage's \texttt{QQbar} (algebraics only), Sage's \texttt{SR} (using GiNaC),
SymPy, Maple, and Mathematica (MMA).
MMA was run on a faster computer in the free Wolfram Cloud, with a 60~s timeout. Other systems were interrupted after $10^3$~s.
Table~\ref{tab:dft} shows select timings;
in most cases, we see order-of-magnitude speedups over
the competing systems.
All timings were done with empty caches;
most systems (including Calcium) run faster a second time,
but comparisons are difficult since
the systems use caches differently.

SymPy fails to prove equality unless $n = 2^k$, and Sage's \texttt{SR} fails except for $n = 2^k, 3, 5, 6, 10, 12, 20$;
Maple (with \texttt{simplify()}) fails on the fourth test sequence for large $n$.
The test case marked (*) only succeeds if we manually disable Gr\"{o}bner bases in Calcium.

{\renewcommand{\arraystretch}{0.9}
\begin{table}
\centering
\caption{Time (s) for exact DFT and zero test.}
\label{tab:dft}
\begin{tabular}{ c c c c c c c c }
$x_{n-2}$  & $N$ & \!\!Sage $\overline{\mathbb{Q}}$\!\! & \!\!Sage \texttt{SR}\!\! & \!\!SymPy\!\! & \!\!Maple\!\! & \!\!MMA\!\! & \!\!Calcium\!\! \\
 \hline
\multirow{4}{*}{$n$} & 6   & 0.020 & 0.047 & fail & 0.016 & 0.078   & 0.00049 \\
  & 8 & 0.033   & 0.11  & 1.1    & 0.0060 &  0.057 & 0.00048 \\
  & 16  & 0.15  & 36    & 9.9    & 0.080 &  0.27    & 0.00068 \\
  & 20  & 0.22  & 124   & fail   & 0.13 &  0.96    & 0.00081 \\
  & 100 & 9.2   & fail  & fail   & 9.1 &  $>60$ & 0.045 \\
\hline
\multirow{4}{*}{$\sqrt{n}$} & 8  &   5.3  &   0.50  &   2.8   & 0.046 &   0.11  &   0.017 \\
 & 16  &   $>10^3$  &   46  &   24  & 0.26 &    0.58   &  0.090 \\
 & 20  &   $>10^3$  &   154 &    fail  & 1.1 &    2.3   &  0.17 \\
 & 100 &    $>10^3$  &   fail  &   fail  & $>10^3$ &    $>60$   &  38 \\
\hline
\multirow{4}{*}{$\log(n)$\!\!\!\!} & 8  &   -  &   0.20  &   1.8   & 0.044 &   0.29  &   0.0059 \\
 & 16  &   -  &   44  &   17 & 0.37 &     0.66  &   0.025 \\
 & 20  &   -  &   136   &  fail  & 0.74 &    45  &   0.046 \\
 & 100  &   - &    fail   &  fail  &  $>10^3$ &   $>60$  &   26 \\
\hline
\multirow{4}{*}{\!\!$e^{2\pi i/n}$\!\!\!\!\!\!} & 8  &   $>10^3$  &   1.3  &   fail  &  0.042 &   0.10  &   0.019 \\
  & 16  &   $>10^3$  &   78   &  $>10^3$  &  0.17 &   0.41  &   0.32 \\
  & 20  &   $>10^3$  &   277   &  fail  &   fail &  $>60$  &   1.1 \\
  & 100 &    $>10^3$  &   fail  &   fail  &  $>10^3$ &   $>60$  &   699* \\
\hline
\multirow{4}{*}{\!\!$\frac{1}{1+n\pi}$\!\!\!\!\!\!} & 8  &   -   &  0.68  &   17  &  0.072 &   0.21  &   0.0041 \\
  & 16  &   -   &  48  &   $>10^3$  &  0.32 &   6.4  &   0.046 \\
  & 20  &   -   &  167   &  fail  &  2.4 &   $>60$  &   0.12 \\
  & 100  &   -  &   fail   &  fail  &  $>10^3$ &   $>60$  &   216 \\
\hline
\multirow{4}{*}{\!\!$\frac{1}{1+\sqrt{n} \pi}$\!\!\!\!\!\!} & 8  &   -  &   0.76  &   22  &  0.074 &   2.6 &    0.082 \\
  & 16  &   -  &   $>10^3$   &  $>10^3$  &  127 &   $>60$  &   8.1 \\
  & 20  &   -  &   fail   &  fail   &  $>10^3$ &  $>60$  &   43 \\
\end{tabular}
\end{table}
}

\subsubsection{Conjugate logarithms}

Example 1 in \cite{BBK2014} asks for simplifying $C_1 = -\frac{1}{8} i \pi \log^{2}(\frac{2}{3} - \frac{2 i}{3}) + \frac{1}{8} i \pi \log^{2}(\frac{2}{3} + \frac{2 i}{3}) + \frac{1}{12} {\pi}^{2} \log(-1 - i) + \frac{1}{12} {\pi}^{2} \log(-1 + i) + \frac{1}{12} {\pi}^{2} \log(\frac{1}{3} - \frac{i}{3}) + \frac{1}{12} {\pi}^{2} \log(\frac{1}{3} + \frac{i}{3})$.
Calcium evaluates this to
$\tfrac{1}{96} (4 \log(\frac{1 - i}{3}) {\pi}^{2} - 8 \log(\frac{-1 - i}{3}) {\pi}^{2} - 5 {\pi}^{3} i)$,
eliminating redundant logarithms.
Calcium does not discover $C_1 = -\tfrac{1}{48} \pi^2 \log(18)$ (the output of Mathematica's \texttt{FullSimplify})
since it does not rewrite the extension numbers, but
it proves equality when this form is given ($C_1 + \tfrac{1}{48} \pi^2 \log(18)$ evaluates to 0).
Calcium takes 0.008~s, or 0.00008~s when the fields are already cached;
Mathematica's \texttt{Simplify} takes 0.015~s
and leaves $C_1$ unsimplified but proves $C_1 + \tfrac{1}{48} \pi^2 \log(18) = 0$ in 0.02~s, while \texttt{FullSimplify} takes 0.03~s.

Maple's \texttt{simplify} returns $-\tfrac{1}{48} \pi^2 (\log(2) + 2 \log(3))$ in 0.000024~s.
Indeed, this is a trivial computation if we (like Maple) atomize
the logarithms.
Since this is not yet implemented in Calcium,
we see the result of the slower, generic approach
using integer relations.

\section{Future work ideas}

We have already discussed several ideas for future work.
Perhaps the most important topics are: new classes
of extension numbers;
simplification, normalization and context-dependent
rewriting of extension numbers; improved numerical algorithms
and methods for working with equivalence classes
of formal fractions; ideal construction.
We conclude by elaborating on the last point.

The usual bottlenecks in constructing ideals (and often in Calcium
as a whole) are: searching for integer relations with LLL,
proving integer relations through recursive computations,
and computing Gr\"{o}bner bases.
Algorithm~\ref{alg:ideal} could be improved in many ways,
most notably through preprocessing
to avoid redundant work and
reduce the dimension or improve numerical
conditioning of LLL matrices.
Some preprocessing strategies are discussed in \cite{BBK2014}.

The PSLQ algorithm is often claimed to be superior to LLL for
integer relations (see for example \cite{BBK2014}),
but this is not obviously true
with modern floating-point LLL implementations.
One benefit of LLL is that we obtain
a matrix of all integer relations at once,
whereas PSLQ has to be run repeatedly to eliminate relations
one by one. We invite further comparison of these algorithms.
In some cases, purely symbolic methods
should be superior to either.

An explicit Gr\"{o}bner basis computation can sometimes be avoided
by setting up extension numbers and relations appropriately.
This is exploited in Magma's algebraically closed field~\cite{Ste2010}.
We have also observed empirically that many calculations in Calcium
work perfectly well (and faster) without computing a Gr\"{o}bner basis,
presumably because the constructed ideal basis is sufficiently triangular
to be effective for reductions (in some cases, we found that
it suffices to compute the Hermite normal form of the LLL output matrices);
a better understanding of this phenomenon would be welcome.

A glaring problem is that when introducing
$n$ extension numbers,
say by adding $\sqrt{2} + \sqrt{3} + \sqrt{5} + \ldots$,
we construct all the intermediate
fields $\mathbb{Q}(a_1)$, $\mathbb{Q}(a_1,a_2)$, \ldots, $\mathbb{Q}(a_1,\ldots,a_n)$,
from scratch.
This is doing nearly $n$ times more work than should be needed.
One possible solution is to let
the user write down a list of extension numbers
and create $\mathbb{Q}(a_1,\ldots,a_n)$ at once for computations.
Another solution is to take advantage of the data that has
already been computed for
$\mathbb{Q}(a_1,\ldots,a_{n-1})$
to generate the data for $\mathbb{Q}(a_1,\ldots,a_n)$.
This seems hard to solve efficiently
and in such a way that the system does not behave
unpredictably depending on the
order of computations.

\bibliographystyle{plain}
\bibliography{references.bib}

\end{document}